%% file: distributedStreamProcessing.tex
\documentclass[10pt, conference, letterpaper]{IEEEtran}

\usepackage{dblfloatfix} % for \figure* at the bottom of a page
\usepackage{balance} % for balancing columns on the last page

\usepackage{amsmath}
\usepackage{amssymb}
\usepackage{amsthm}
\usepackage[noend]{algorithmic}
\usepackage{algorithm}
\usepackage{cases}
\usepackage{comment}
\usepackage{xcolor}
\usepackage{graphicx,epstopdf} % needed for eps conversion in pdflatex
\usepackage{multicol}
\usepackage[T1]{fontenc}

\newcommand{\scomp}{\ensuremath{\circ_s}}
\newcommand{\pcomp}{\ensuremath{\circ_p}}
\newcommand{\BigO}{\mathcal{O}}
\renewcommand{\qed}{\hfill\ensuremath{\square}}

\newtheorem{theorem}{Theorem}[section]
\newtheorem{lemma}[theorem]{Lemma}
\newtheorem{fact}[theorem]{Fact}

\newtheorem{definition}[theorem]{Definition}
\newtheorem{problem}[theorem]{Problem}

\begin{document}

\sloppy

\title{Task Allocation for Distributed Stream Processing\\(Extended Version)}
\date{}

\author{Raphael Eidenbenz \\ABB Corporate Research %\\ Baden-D\"attwil, Switzerland
\\ raphael.eidenbenz@ch.abb.com 
\and
Thomas Locher \\ABB Corporate Research %\\ Baden-D\"attwil, Switzerland
\\thomas.locher@ch.abb.com
}

\maketitle

\input{abstract}

%\begin{center}
%{ \it Regular paper submission to DISC 2015.  Not eligiblLaisssssssssse for the best student paper award.}
%\end{center}

%\category{F.2.2}{Analysis of Algorithms and Problem Complexity}{Nonnumerical Algorithms and Problems}[computations on discrete structures]
%A category including the fourth, optional field follows...
%\category{G.2.2}{Discrete Mathematics}{Graph Theory}[network problems]


%\keywords{Stream processing, task allocation, series-parallel graphs, approximation algorithms}

%\end{titlepage}
%\newpage

\input{introduction}

\input{model}

\input{algorithm}

%\input{discussion}

\input{relatedwork}
\input{conclusion}

%\footnotesize

\balance
\bibliographystyle{IEEEtran}
\bibliography{distributedStreamProcessing}

%\newpage
%\normalsize
%\appendix
%\input{appendix}

%\newpage
%\normalsize
%\appendix
%\input{appendix}

\end{document}

%% file: abstract.tex
% !TeX root = distributedStreamProcessing.tex
\begin{abstract}

There is a growing demand for live, on-the-fly processing of increasingly large amounts of data.
In order to ensure the timely and reliable processing of streaming data, a variety of distributed stream processing architectures and platforms have been developed, which handle the fundamental tasks of (dynamically) assigning processing tasks to the currently available physical resources and routing streaming data between these resources.
However, while there are plenty of platforms offering such functionality, the theory behind it is not well understood. In particular, it is unclear how to best allocate the processing tasks to the given resources.

In this paper, we establish a theoretical foundation by formally defining a task allocation problem for distributed stream processing, which we prove to be NP-hard. Furthermore, we propose an approximation algorithm for the class of
series-parallel decomposable graphs,
which captures a broad range of common stream processing applications. The algorithm achieves a constant-factor approximation under the assumptions that the number of resources scales at least logarithmically with the number of computational tasks and the computational cost of the tasks dominates the cost of communication.

\begin{comment}
In this paper, we establish a theoretical foundation by formally defining a task allocation problem for distributed stream processing, which we prove to be NP-hard, even when restricted to a class of series-parallel graphs that arguably captures a broad range of common stream processing applications. Furthermore, we propose an approximation algorithm tailored to this graph class that achieves a constant-factor approximation, under the assumptions that the number of resources scales at least logarithmically with the number of computational tasks and the computational cost of the tasks dominates the cost of communication.
\end{comment}
\end{abstract}
% tailored to a specific class of series-parallel graphs 

%% file: introduction.tex
% !TeX root = distributedStreamProcessing.tex
\section{Introduction}
\label{sec:introduction}

% Applications:
The stream processing paradigm, where data streams are processed by applying a series of functions to the elements in the data streams, is gaining in importance due to the large variety of supported applications.
An early adopter of stream processing, and the related complex event processing paradigm, was the financial service industry, where it is used, e.g., 
to rapidly detect relevant events in high-frequency stock trading. Stream processing is also applied to digital control systems in order to continuously supervise and record the system state, and in network security to monitor network traffic.
Another use case is customer experience management, where, for example, click stream data is analyzed on-the-fly to measure customer behavior.
As the amount of data to process grows and the steady, uninterrupted processing of data becomes more and more critical,
scalability and fault-tolerance are becoming key requirements. Distributed stream processing platforms address these issues by spreading the workload across an extensible network of machines and by dynamically redistributing tasks in the event of machine or network failures.
Stream processing jobs typically exhibit the properties that individual data items can enter the various stages of processing independently and simultaneously and that the number of arithmetic operations per data transfer is high, i.e., the computational complexity dominates the communication cost.

The rising popularity of distributed stream processing has led to the development of many systems (e.g., \cite{akidau13,arasu03,chandrasekaran03,gedik08,neumeyer10}) that offer simple programming interfaces while abstracting away the underlying complexity of distributing tasks, routing streams, and handling failures, akin to the MapReduce framework~\cite{dean04} for batch processing. 
The basic principle behind these platforms is the notion of a processing element (PE) that consumes one or more data streams, processes the received data, and continuously outputs processing results again in the form of data streams. 
These PEs are interconnected, forming a \emph{streaming topology} where PEs without incoming streams receive external data streams for processing, and the set of PEs without outgoing streams constitutes the final stage of processing after which the results are typically stored or displayed.
While a PE technically encapsulates a specific computational task, we treat the terms \emph{PE} and \emph{task} as synonymous.
An essential function of a streaming platform is to allocate the PEs to the available physical resources. 
In the following, we will refer to this function as \emph{task allocation}. In most proposed systems, simple schemes, such as round-robin, and heuristics are used to allocate tasks. The lack of a formal specification of the problem negates the possibility to optimize the distribution of PEs for an optimal utilization of the resources, which would result in a minimized processing delay.

In this paper, we focus on this task allocation problem for distributed stream processing and propose a basis for a formal treatment, enabling the rigorous analysis of allocation strategies.
Furthemore, we introduce a novel class of series-parallel-type graphs that we use to model streaming topologies: This model is based on the premise that many processing jobs either consist of a series of linearly dependent tasks, tasks that can be executed independently in parallel, or a combination thereof. We argue that this model is of particular interest as it adequately captures a large subset of practical stream processing applications.
The task allocation problem is shown to be NP-hard, even when the streaming topology is restricted to the aforementioned class of series-parallel graphs. The main contribution of the paper is an approximation algorithm that deterministically achieves an approximation ratio of $\BigO(n^{\BigO(1/c)})$, where $n$ denotes the number of tasks and $c$ is the number of physical resources---subject to the constraint that the communication cost is upper bounded asymptotically by the computational cost, which, as stated before, is often the case in practice.
Hence, the algorithm guarantees a constant-factor approximation to the cost of the best possible allocation provided that $c \in \Omega(\log n)$.

The paper is structured as follows. In the next section, the model and a few general complexity results are presented. The approximation algorithm is described and analyzed in Section~\ref{sec:algorithm}. 
%The difficulty of handling large communication costs  and generalizations of the task allocation problem are discussed in Section~\ref{sec:discussion}. 
Related work on task allocation and distributed stream processing is summarized in Section~\ref{sec:relatedwork}. Finally, Section~\ref{sec:conclusion} concludes the paper.

%% file: model.tex
% !TeX root = distributedStreamProcessing.tex

\section{Model}
\label{sec:model}

In this section, we introduce the task allocation problem and the class of {series-parallel-decomposable} graphs, and present some general complexity bounds.

\subsection{Task Allocation}
We model the streaming topology as a directed acyclic \emph{streaming graph} $G = (V,E)$, where $|V| =: n$.
Let $V_s \subseteq V$ denote \emph{source nodes}, i.e., the set of nodes without incoming edges, and let $V_t \subseteq V$ be the \emph{sink nodes}, i.e., the set of nodes without outgoing edges.
Graph $G$ represents the stream processing topology in that each node $v\in V$ corresponds to a processing task, where a \emph{task} is an atomic unit of computation to be executed on a single machine.
A weight function $w: V \rightarrow \mathbb{R}^+$ determines the computional complexity of the tasks, e.g., reflecting the time it takes to process an item in the data stream.
This definition assumes that the task complexity is independent of the particular data item to process, i.e., the computational effort is equal for all items in the data stream. While this is a simplification, the variance is small for a wide range of tasks in practice, e.g., processing a steady stream of measurements over a fixed size window, and can often be neglected.

The tasks must be allocated to $c$ processing resources $\mathcal{R} := \{R_1,\ldots,R_c\}$, which is formally captured by a function $r:V \rightarrow \mathcal{R}$ that defines the allocation of tasks to resources.
The number of tasks mapped to resource $R$ is defined as $n(R) := |\{v \in V \,|\, r(v) = R\}|$. This definition implies that $\sum_{i=1}^c n(R_i) = n$.
The more tasks are executed on the same resource $R$, the more it must divide its processing capacity. We define the \emph{processing cost} of a task $v$ as
\begin{displaymath}
d(v) := w(v)\cdot n(r(v)),
\end{displaymath}
i.e., the cost grows linearly with the weight and the number of tasks mapped to the same resource. 
Such a linear cost model reflects the concurrency model where each allocated task gets an equal share of the resource, which entails that the individual processing times grow linearly with the number of collocated tasks. Note that we refrain from introducing a scaling parameter to the cost of collocating tasks but use the number $n(r(v))$ of collocated tasks directly to compute the processing cost of a task, i.e., we assume that weights are scaled appropriately.
The advantage of collocating tasks on the same resource is that data streams between those tasks need not be routed over the network, an operation which incurs a certain cost.
The cost of transferring data along an edge $e = (u,v) \in E$ is given by its edge weight $b(e) \ge 0$, which must be paid only if $u$ and $v$ are not collocated, i.e., communication on the same resource is assumed to be free of cost. More formally, the \emph{transfer cost} $\ell(e)$ of $e=(u,v)$ is
\begin{displaymath}
 \ell(e) := \begin{cases} b(e) &\mbox{if } r(u) \neq r(v) \\ 
0 & \mbox{otherwise.}
\end{cases}
\end{displaymath}
Transfer costs can be used to model data rates, latencies, or a combination thereof.
The \emph{streaming cost} on a (directed) $v_s$-$v_t$-path $P = (v_s,v_1,\ldots,v_\ell,v_t)$, where $v_s \in V_s$ and $v_t \in V_t$, is defined as the sum of the processing cost of each node on the path, plus the transfer cost of each edge on the path.
\begin{displaymath}
d(P) := \sum_{v \in P} d(v) + \sum_{e\in P} \ell(e).
\end{displaymath}

Let $\mathcal{P}(G)$ denote the set of all paths in $G$ starting at a source node and terminating at a sink node.
The primary objective is to minimize the streaming cost of the entire graph $G$, which is defined as the streaming cost on the worst-case path, i.e.,
\begin{displaymath}
d(G) := \max_{P\in \mathcal{P}(G)} d(P).
\end{displaymath}

\begin{problem}[Task Allocation]\label{def:general_problem} 
Given a weighted directed acyclic graph $G$ and a set $\mathcal{R}$ of $c$ resources, find the mapping $r: V \rightarrow \mathcal{R}$ that minimizes $d(G)$.
\end{problem}

The relation between node and edge weights, and the resulting processing and transfer costs, has an immediate impact on the streaming cost.
As mentioned before, it is often reasonable to assume that processing costs dominate the transfer costs, %the edge weights are upper bounded by the node weights, 
which reflects scenarios where the resources are in physical proximity and connected by means of high-bandwidth, low-latency links. We formally define such computationally constrained problems in Definition~\ref{def:compConstrained}. Let $\hat{d}(G)$ be the streaming cost of $G$ when setting $\ell(e) := 0$ for all edges.

\begin{definition}\label{def:compConstrained}
A streaming graph is \emph{computationally constrained} if $d(G) \in \BigO(\hat{d}(G))$ for any  mapping $r: V \rightarrow \mathcal{R}$.
\end{definition}

In this paper, we primarily study such problem instances and discuss implications for the general case along the way.

\subsection{Series-Parallel-Decomposable Graphs}

Regarding streaming topologies, we focus our attention on directed \emph{series-parallel-decomposable (SPD) graphs}, which are graphs that can be constructed by a combination of serial and parallel composition steps. 

A \emph{serial composition} $G = G_1 \scomp G_2$ of two SPD graphs $G_1$ and $G_2$ is defined as follows: The resulting graph $G$ consists of both graphs $G_1$ and $G_2$ where each sink node of $G_1$ is connected to each source node of $G_2$. More formally, $G=(V,E)$ is given by
%\begin{eqnarray*}
$V := V_1 \cup V_2$ and %\text{~~and~~}
 $E :=  E_1 \cup E_2 \cup \{(u,v) ~|~ u\in V_{t,1} \text{~and~} v\in V_{s,2}\}$,
%\end{eqnarray*}
where $V_{s,i}$ and $V_{t,i}$ are the source and sink nodes of $G_i$, respectively.
A \emph{parallel composition} $G = G_1 \pcomp G_2$ of two SPD graphs $G_1$ and $G_2$ is a mere union of the two graphs without adding any edges, i.e.,
$V:= V_1 \cup V_2$ and $ E :=  E_1 \cup E_2$.
Note that all source nodes remain sources, and all sink nodes remain sinks.
Formally, SPD graphs are defined as follows.
\begin{definition}[Series-Parallel-Decomposable]\label{def:spdgraphs} 
A directed graph $G=(V,E)$ is series-parallel-decomposable (SPD) if $|V| = 1$ or there are SPD graphs $G_1$ and $G_2$ such that
$G = G_1 \scomp G_2$ or $G = G_1 \pcomp G_2$.
\end{definition}

\begin{comment}
\begin{definition}[Series-Parallel-Decomposable (SPD)]\label{def:spdgraphs} 
A directed graph $G=(V,E)$ is series-parallel-decomposable if and only if 
\begin{itemize}
\item it can be constructed by a serial composition of SPD graphs,
\item it can be constructed by a parallel composition of SPD graphs, or
\item $|V|=1$.
\end{itemize}
\end{definition}
\end{comment}

Since SPD graphs are defined recursively, they can be represented by a rooted series-parallel decomposition tree (\emph{SPD tree}): The leaves of the SPD tree correspond to the nodes of the SPD graph. Each internal node of the tree represents a serial or parallel composition, i.e., a subtree $T$ with root $r$ corresponds to the graph that results from the composition of the graphs corresponding to the subtrees rooted at $r$'s child nodes. All internal nodes are labeled with $s$ or $p$ to indicate the type of the composition, i.e., serial or parallel. Given an SPD tree, the corresponding SPD graph can be constructed in linear time by traversing the tree in post-order, i.e., from the leaves to the root. % (cf Algorithm~\ref{algo:constructSPDGraph}).
An example SPD graph is shown in Figure~\ref{fig:series-parallel}, and its SPD tree is depicted in Figure~\ref{fig:series-parallel-recursive}.

\begin{figure}[t]
\center
 \includegraphics[width=.9\columnwidth]{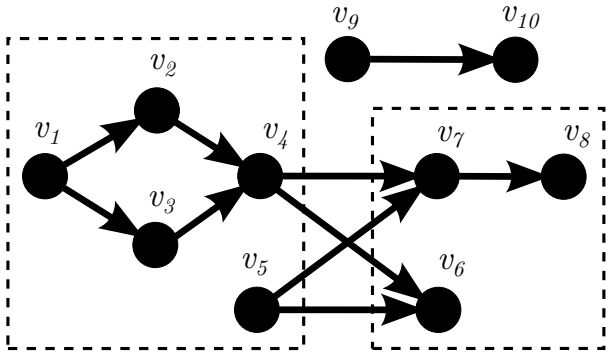}
 \caption{An example SPD graph consisting of 10 nodes.}
 \label{fig:series-parallel}
\end{figure}

In the remainder, we will often describe streaming graphs by its compositional structure and use the SPD tree representation for recursive algorithms.
It is convenient to identify a component $G_i$ of an SPD graph by the corresponding node in the SPD tree representation, i.e., the root of the sub-tree $T(G_i)$ corresponding to $G_i$. For ease of notation, we describe an internal node $z$ of an SPD tree by the type of composition and its child nodes, i.e., $z=(op,\{z_1,z_2\})$ where  $op \in \{s,p\}$ and $z_1, z_2$ are the child nodes. Moreover, we extend the SPD tree representation in that we allow nodes to have more than two children, which enables the concise representation of concatenations of serial or parallel compositions: a graph $G = G_1 \scomp G_2 \scomp \ldots \scomp G_k$ can be represented by an SPD tree $T(G)$ with a root node $z=(s,\{z_1, \ldots, z_k\})$ that has $k$ children $\{z_1,\ldots, z_k\}$, each of which is the root node of a subtree $T(G_i)$. Accordingly, $z=(p,\{z_1, \ldots, z_k\})$ represents $k$ parallel components. Finally, let $\mathcal{C}(z)$ denote the set of $z$'s  children. If $z$ is a leaf node, then $\mathcal{C}(z) = \emptyset$.

Note that the class of SPD graphs does not coincide with the class of series-parallel graphs as defined by Takamizawa et al.~\cite{takamizawa82}. While many graphs are both series-parallel and SPD, there are graphs that are only in one of the two classes.
For example, the graph with node set $V=\{v_1,v_2,v_3\}$ and edges $E=\{ (v_1,v_2), (v_2,v_3), (v_1,v_3)\}$ is series-parallel but not SPD. There are many SPD graphs that are non-planar, e.g., the class of SPD graphs contains all complete bipartite graphs $K_{m,n}$, which are non-planar if $m\ge n \ge 3$. In contrast, series-parallel graphs are planar by design.

\begin{figure}[t]
\center
\includegraphics[width=.9\columnwidth]{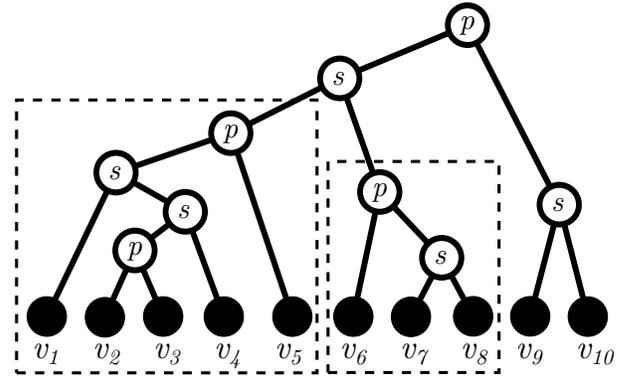}
 \caption{Tree representation $T(G)$ of the SPD graph~$G$ in Figure~\ref{fig:series-parallel}.}
 \label{fig:series-parallel-recursive}
\end{figure}

\subsection{General Bounds}\label{sec:general_bounds}

Allocating the tasks of a streaming graph to a set of resources in an optimal fashion is a hard problem in general.

\begin{theorem} The task allocation problem is NP-hard. \label{thm:NPhardGeneral}
\end{theorem}
\begin{proof}
This result follows from a polynomial reduction from the NP-complete  {\tt SUBSETSUM} problem, which, given a multiset $\mathcal{S} := \{s_1,\ldots,s_n\}$ of positive integers and an integer $x$, asks for a subset $\mathcal{S}^*$ of $\mathcal{S}$ such that the sum of the numbers in $\mathcal{S}^*$ equals $x$. 
%For convenience, we assume that $\mathcal{S}$ is sorted in descending order. Moreover, 
We assume w.l.o.g.\ that $x\ge s_i$ for all $i$. Note that if $\mathcal{S}$ contained any $s_i > x$, it could be immediately ruled out from the solution.
Given an instance $(\mathcal{S},x)$ of {\tt SUBSETSUM}, for any $k \in \{1,\ldots n\}$ we construct a task allocation problem instance $I^k=(G^k,w^k,b^k,\mathcal{R}^k)$. The optimal mapping of $G^k$ to $n+k$ resources yields a subset $\mathcal{S}(I^k) \subseteq \mathcal{S}$ as described below. The claim is that if $(\mathcal{S},x)$ has a solution consisting of $k$ elements, then $\mathcal{S}(I^k)$ is such a solution. Hence, solving  all $I^k$ for $k \in \{1,\ldots n\}$ either reveals a solution---a candidate solution can be checked in polynomial time---or answers the subset sum problem in the negative. 

\begin{figure*}[t]
\center
 \includegraphics[width=.8\textwidth]{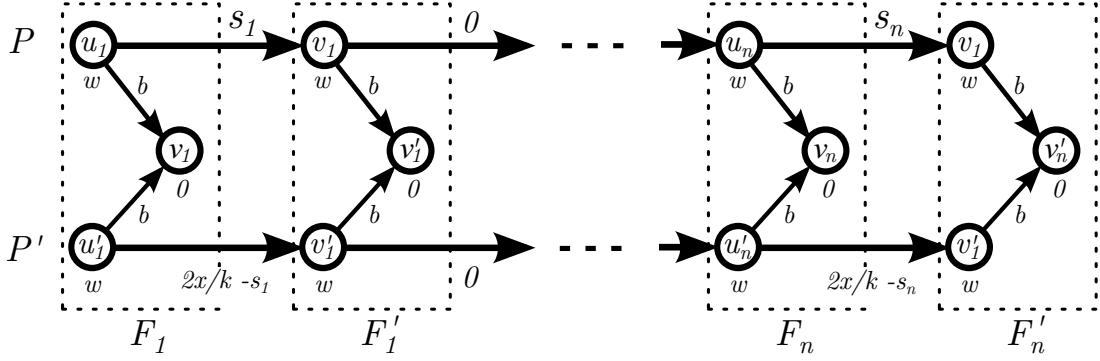}
 \caption{The task allocation instance constructed in the reduction from  {\tt SUBSETSUM}. Here, $b=12nw+x$.}
 \label{fig:np-proof}
\end{figure*}

An instance $I^k$ is constructed as follows (see Figure~\ref{fig:np-proof}):
$G^k$ consists of two paths $P:=u_1,v_1,u_2,v_2, \ldots, u_n, v_n$ and $P' :=u'_1,v'_1,u'_2,v'_2, \ldots, u'_n, v'_n$, and $2n$ additional nodes $\nu_1,\nu'_1, \ldots,\nu_n, \nu'_n$. Each node $\nu_i$ has incoming edges from $u_i$ and $u'_i$, and each node $\nu'_i$ has incoming edges from $v_i$ and $v'_i$. 
Edge weights $b^k$ are given by $b^k(u_i,v_i) := s_i$, $b^k(u'_i,v'_i) := (2x/k) - s_i$, $b^k(v_i,u_{i+1})=b^k(v'_i,u'_{i+1}) := 0$, and $b^k(e):=12nw+x$ for all remaining edges $e$.
Node weights $w^k$ are $0$ for all nodes $\nu_i$, $\nu'_i$, and $w:= \sum_{i=1}^n s_i$ for all other nodes.
Finally, $\mathcal{R}^k$ consists of $n+k$ resources.
Note that only $b^k$ and $\mathcal{R}^k$ depend on $k$, whereas $G^k$ and $w^k$ remain the same.
Given a solution $r$ of $I^k$, the {\tt SUBSETSUM}  solution candidate $\mathcal{S}(I^k)$ is defined as the set of all $s_i=b(u_i,v_i)$  for which $r(u_i)\neq r(v_i)$. 

It remains to show that if there is a {\tt SUBSETSUM} solution $\mathcal{S}^*$ with $k$ elements, then the optimal allocation $r^*$ for $I^k$ yields a correct solution, i.e.,  $\mathcal{S}(I^k)=\mathcal{S}^*$.
Let $F_i := \{u_i,u'_i,\nu_i\}$ and $F'_i := \{v_i,v'_i,\nu'_i\}$. Note that $V$ consists of $2n$ such node triplets, which we will refer to as \emph{forks}.
Assuming there is a solution with $k$ elements, then the claimed optimal allocation $r^*$ maps  forks $F_i$ and $F'_i$ to a separate resource each if $s_i \in \mathcal{S}^*$ and collocates all nodes in $F_i \cup F'_i$ otherwise; thus, $2k$ resources host one fork, $n-k$ resources host two forks.  The resulting streaming cost $d^*$ is given by the maximum over the two paths along $P$ and $P'$, both ending in $\nu'_n$, and amounts to
\begin{eqnarray*}
d^*(P) & = & 2k\cdot 3w + 2(n-k)\cdot 6w + \sum_{s_i \in \mathcal{S}^*} s_i , \\
d^*(P') & = & 2k\cdot 3w + 2(n-k)\cdot 6w + \sum_{s_i \in \mathcal{S}^*} \left( \frac{2x}{k} - s_i \right)  , \\
d^*(G^k) & = & \max \{d^*(P), d^*(P')\} = (12n-6k) w + x ,
\end{eqnarray*}
%\begin{eqnarray*}
%&d^*(P) = 2k\cdot 3w + 2(n-k)\cdot 6w + \sum_{s_i \in \mathcal{S}^*} s_i , ~~~~ d^*(P') = 2k\cdot 3w + 2(n-k)\cdot 6w + \sum_{s_i \in \mathcal{S}^*} \left( \frac{2x}{k} - s_i \right),\\
%&d^*(G^k)  =  \max \{d^*(P'), d^*(P')\} = (12n-6k) w + x , \text{~~ since~~} |\mathcal{S}^*|=k \text{~~ and~~}  \sum_{s_i \in \mathcal{S}^*} s_i = x.
%\end{eqnarray*}
%where the last equation holds since $|\mathcal{S}^*|=k$ and  $\sum_{s_i \in \mathcal{S}^*} s_i = x$.
In the following, we show that $r^*$ is indeed optimal.
%
%P1: All forks are collocated
%{\bf P1.} 
In an optimal allocation, all nodes of a fork must be collocated, otherwise the cost $(12nw+x)$ of an edge to $\nu_i$ or $\nu'_i$ has to be paid on one path ending either in $\nu_i$ or $\nu'_i$. The cost on that path exceeds $d^*$, contradicting optimality.
%
%{\bf P2.} 
For a large enough $w$, a most even distribution of the forks onto resources is optimal since the processing cost of a set of collocated forks grows quadratically with the size of the set: the processing cost on path $P$ and $P'$ for $\kappa$ collocated forks is $\kappa^2\cdot 3w$. Hence, $2k$ resources must contain exactly one fork, and $n-k$ resources must contain exactly two forks in an optimal allocation. The chosen $w= \sum_{i=1}^n s_i$ is large enough so that any deviating allocation costs more than $r^*$.
Consequently, at most $n-k$ edges $(u_i,v_i)$ with weights $s_i$ can be \emph{covered}, i.e., $r(u_i)=r(v_i)$, and $2k$ edges $(u_i,v_i)$ on $P$ must be left \emph{uncovered} ($r(u_i) \neq r(v_i)$). Note that if edge $(u_i,v_i)$ is covered on $P$, then $(u'_i,v'_i)$ is covered on $P'$ as well. It is optimal to leave all $k$ edges with weight $0$ and a set $S$ of $k$ edges with $s_i$ weights uncovered on $P$, and a corresponding set $S'$ on $P'$.

If $\sum_{e \in S} b(e) > x$, then the streaming cost of $P$ exceeds $d^*$. 
On the other hand, if $\sum_{e \in S} b(e) < x$, then 
$$\sum_{e \in S'} b(e) = \sum_{e \in S'} \frac{2x}{k} - \sum_{e \in S} b(e)  = 2x - \sum_{e \in S} b(e) > x,$$
and $d(P') > d^*$. Hence, $S$ must be chosen so that $\sum_{e \in S} b(e)=x$. Allocation $r^*$ is optimal and $\mathcal{S}(I^k)=\mathcal{S}^*$.
%and the $s_i$ weights on the uncovered edges in $S$ correspond to a solution $\mathcal{S}^*$ of the subsetsum problem.
\flushright\vspace{-3mm}\qedhere
\end{proof}

The proof of Theorem~\ref{thm:NPhardGeneral} shows NP-hardness with a streaming graph that is not quite an SPD graph due to the edges ending at nodes $\nu_i$.
Moreover, it achieves the reduction from the subset sum problem using a construction that relies on variable transfer and task weights.
Since many hard problems can be solved efficiently on series-parallel graphs~\cite{takamizawa82} an interesting question is whether optimal task allocation is efficiently solvable on SPD graphs. The following theorem answers this question in the negative, even for computationally constrained problems and even if all transfer weights and all task weights are constant. The claim follows from a reduction from the decision version of the  NP-complete  \emph{partition problem}. 
%The reader is referred to Appendix~\ref{app:NPharduniform} for the proof and some additional remarks.

\begin{theorem} \label{thm:NPharduniform}
The task allocation problem on SPD graphs in concise format is NP-hard, even if transfer and task weights are uniform.
\end{theorem}
\begin{proof}
The claim follows from a reduction from the decision version of the  NP-complete  \emph{partition problem}, which, given a multiset $\mathcal{S} := \{s_1,\ldots,s_n\}$ of positive integers, asks for a partitioning of $\mathcal{S}$ into two subsets $\mathcal{S}_1$ and $\mathcal{S}_2$ such that the sum of the numbers in $\mathcal{S}_1$ is equal to the sum of the numbers in $\mathcal{S}_2$. Given an instance $\mathcal{S}=\{s_1,\ldots,s_k\}$ of the partitioning problem, we construct a task allocation problem instance as follows: for each $s_i$, we construct a streaming graph component $G_i$ consisting of a serial composition of two sets of $s_i$ parallel nodes, i.e., $G_i :=G_{i,l} \scomp G_{i,r}$ where $G_{i,l} =G_{i,r}=(V,E), |V|=s_i$ and $E=\emptyset$. Edge weights and node weights have positive constant values, i.e., $w(v)=\alpha$, $b(e)=\beta$ for all $v\in V$, $e\in E$. Let $\alpha=\beta:=1$ for simplicity. 
Let $d$ be the output of the algorithm that computes the streaming cost of the graph $G:=G_1 \pcomp G_2 \pcomp \ldots \pcomp G_k$ given that $\mathcal{R}=\{R_1,R_2\}$. If $d=n:=2\cdot \sum_{i=1}^k s_i$, we answer the partitioning problem positively and negatively otherwise.
It remains to show that $d(G)=n$ if and only if there exists a perfect partitioning of $\mathcal{S}$. If a perfect partitioning exists then the following mapping provides an optimal task allocation:
\begin{displaymath}
 r(v) := \begin{cases} R_1 &\mbox{if } v\in G_i \mbox{ and } \sigma(s_i)=S_1 \\ 
R_2 & \mbox{otherwise,}
\end{cases}
\end{displaymath}
where $\sigma$ is an optimal partitioning of $\mathcal{S}$ into two sets $S_1,S_2$. Since $\sigma$ is perfect it holds that $|R_1|=|R_2|=n/2$. Each path in $G$ entails processing cost of $2\cdot n/2$ and zero transfer cost. For any other mapping $r'$ with $n(R_1)=n(R_2)=n/2$ there is at least one path $(u,v)$ with $r'(u)\neq r'(v)$ yielding streaming cost $d(G)>n$. For any mapping $r'$ with $n(R_1)\neq n(R_2)$, let  $n(R_1)> n(R_2)$ w.l.o.g., then there exists a component $G_i$ with more nodes mapped to $R_1$ than $R_2$. Hence, there exists a path $(u,v)\in G_i$ with $r'(u)=r'(v)=R_1$ yielding streaming cost $d(G)=2\cdot |R_1|>n$, which contradicts optimality.
 If no perfect partitioning exists then any mapping $r$ with $n(R_1)=n(R_2)=n/2$ implies that there is a component $G_i$ with nodes mapped to $R_1$ and $R_2$. Hence, there exists a path in $G_i$ with cost $2\cdot n/2 + 1$ and $d(G)>n$. If the optimal mapping $r$ chooses $n(R_1)\neq n(R_2)$ then the same argument holds as in the case where a perfect partitioning exists: if  $n(R_1)> n(R_2)$, then there is a $P=\{(u,v)\}$ with $r(u)=r(v)=R_1$, which entails that $d(G)>n$.
\end{proof}

Note that the reduction of the proof of Theorem~\ref{thm:NPharduniform} uses a streaming graph where the number of nodes is proportional to the sum of the numbers $s_i$ of the partitioning problem. Since partition problem instances are only NP-hard if they contain $s_i \in \mathcal{S}$ that are exponentially large in $|\mathcal{S}|$~\cite{mertens06}, the used streaming graph contains a number of nodes that is exponential in the bit representation of the partitioning instance. %, which is in $\BigO(|\mathcal{S}|\cdot \log \max_{i} s_i)$. 
Therefore, Theorem~\ref{thm:NPharduniform} proves NP-hardness only for concise representations of task allocation problem instances. For example, each $G_i$ component of the used graph can be described in polynomial space similarly as in the proof. While we leave the question of hardness for non-concise SPD graph representations open, it is possible to adapt the approximation algorithm presented in Section~\ref{sec:algorithm} so as to handle concise instances of the graphs used in the proof of Theorem~\ref{thm:NPharduniform}.
Note that the edge and task weight constants, $\alpha$ and $\beta$, in the proof of Theorem~\ref{thm:NPharduniform} can be set independently to any positive value. As such, the proof  holds for computationally constrained graphs as well as non-constrained graphs.

While the task allocation problem is NP-hard, the simple algorithm that assigns all tasks to one resource, i.e., $r(u):=R_1$ for all $u \in V$, achieves an $n$-approximation: Obviously, we have that $\ell(e) = 0$ for all $e\in E$, which implies that $d(P) = \sum_{v \in P} d(v) = n\sum_{v \in P} w(v)$ for all $P \in \mathcal{P}$. Therefore, we get that $d(G) = n \sum_{v \in P'} w(v)$ where $P'$ is the path with the largest sum of node weights.
Since for \emph{any} allocation, $d(G) \ge d(P')$ and the minimum cost of path $P'$ is
\begin{displaymath}
d(P') =  \sum_{v \in P'} d(v) + \sum_{e \in P'} \ell(e) \ge \sum_{v \in P'} d(v) \ge \sum_{v \in P'} w(v),
\end{displaymath}
the claimed bound follows.
This straightforward solution is optimal if $c=1$ or if edge weights are exceedingly large; in particular if $b(e) \ge k n D w_{max}$ for each edge $e\in E$, where $D$ is the graph diameter, $w_{max} := \max_{v \in V} w(v)$, and $k$ is a constant $\ge 1$. Collocating all nodes on one resource is optimal since the streaming cost of any path is upper bounded by $nDw_{max}$.

The case of large edge weights can be considered the opposite of computationally constrained graphs, since the transfer costs dominate the solution rather than the processing costs. Imposing a specific \emph{upper bound} on the edge weights, on the other hand, results in a computationally constrained graph:

\begin{lemma}\label{lma:compConstrained}
If $b(e) \le k w_{min}\lceil D/c\rceil$ for each $e\in E$, where $w_{min} := \min_{v \in V} w(v)$, $D$ is the diameter of $G$,  and $k$ is a positive constant, then $G$ is computationally constrained.
\end{lemma}
\begin{proof}
 Due to the bound on $b(e)$, the maximum transfer cost along any path is $\max_{P \in \mathcal{P}(G)} \sum_{e \in P} b(e) \le D k w_{min}\lceil D/c\rceil$.
We will now show that
$\hat{d}(G) \in \Omega\left(w_{min}\left(D+D^2/c\right)\right)$
for any $G$, $c$, and mapping $r$, which implies that $d(G) \in O(\hat{d}(G))$.
The lower bound of $\Omega(w_{min}D)$ is trivial since $d(v) \ge w_{min}$ for all $v \in V$, resulting in a total processing cost of at least $w_{min}D$ on each path of length $D$.
Assume for the sake of argument that there is a path $P$ of length $D$ where $w(v) = w_{min}$ for each $v\in P$. The streaming cost of $P$ is at least $\lceil w_{min}D^2/c\rceil$ even when excluding transfer cost and \emph{all} $c$ resources are used up exclusively for $P$. Naturally, the streaming cost can only increase when there are nodes with larger weights, the resources are shared with other nodes, or the nodes on $P$ are mapped to fewer than $c$ resources. Thus, $\Omega(w_{min}D^2/c)$ is a lower bound on $d(G)$ as claimed.
\end{proof}

The lemma shows that the edge weights may be considerably larger than the node weights, in the order of $D/c$, without affecting the asymptotic streaming cost. More generally, if $\sum_{e \in P} b(e) \in O(\hat{d}(G))$ for all $P \in \mathcal{P}(G)$ for any mapping $r: V \rightarrow \mathcal{R}$, then $G$ is computationally constrained as well.

%% file: algorithm.tex
% !TeX root = distributedStreamProcessing.tex
\section{Algorithm}
\label{sec:algorithm}
Before describing the main algorithm, we start with a straightforward algorithm to illustrate the difficulty of the task allocation problem.
The algorithm $\mathcal{A}^{avg}$ strives to distribute the work equally among the resources. Specifically, it partitions the nodes into $c$ groups such that the sum of node weights is as balanced as possible. More formally, it minimizes
\begin{displaymath}
\max_{R\in \mathcal{R}} \sum_{v \in r^{-1}(R)} w(v) -\min_{R\in \mathcal{R}} \sum_{v \in r^{-1}(R)} w(v),
\end{displaymath}
where $r^{-1}: \mathcal{R} \rightarrow 2^V$ determines the set of nodes mapped to a particular resource $R \in \mathcal{R}$.\footnote{Note that finding such a partitioning is NP-hard by itself.}  
While this approach seems reasonable, there are instances where $\mathcal{A}^{avg}$ fails to achieve a better approximation ratio than the trivial algorithm that only uses one resource. We can take the graph consisting of $n$ parallel nodes, where $w(v_1) := n/3$ and $w(v_i) := 1$ for all $i\in \{2,\ldots,n\}$, as an example and set $c:=2$. As the sum of weights is $\frac{4}{3}n-1$, algorithm $\mathcal{A}^{avg}$ attempts to assign $\frac{2}{3}n-\frac{1}{2}$ work to each resource. Without loss of generality, let $r(v_1) = R_1$. Since the weight of all other nodes is $1$, we get that $n(R_1) \ge \frac{1}{3}n-1$. Hence it follows that $d(v_1) \ge \frac{n^2}{9}-\frac{n}{3} \in \Omega(n^2)$, implying that $d(G) \in \Omega(n^2)$. The optimal solution, however, dedicates one resource completely to $v_1$, which entails that $d(v_1) = n/3$ and $d(v_i) = n-1$ for all $i\in \{2,\ldots,n\}$. The streaming cost is therefore only $d(G) \in \BigO(n)$.

Instead of tackling the task allocation problem directly, we will now take a detour and present an algorithm for a continuous version of the problem, which our main algorithm will leverage.

\subsection{Continuous Algorithm}

The continuous version of Problem~\ref{def:general_problem} is defined as follows. There is a \emph{capacity} $c > 0$ that must be assigned to the tasks, i.e., each task $v$ gets a \emph{share} $x(v) \in \mathbb{R}^+$ of the capacity subject to $\sum_{v \in V} x(v) \le c$. Given a task's weight and share, the \emph{continuous processing cost} $\delta(v)$ is defined as
%\begin{displaymath}
$\delta(v) := w(v) / x(v)$. 
%\end{displaymath}
There are no transfer costs in this model, and hence the streaming cost of a path $P$ is $\delta(P) := \sum_{v\in P} \delta(v)$. As in the discrete model, the streaming cost of $G$ in the continuous model is the maximum streaming cost over all paths, i.e.,
%\begin{displaymath}
$\delta(G) := \max_{P\in \mathcal{P}(G)} \delta(P)$.
%\end{displaymath}
%
The goal is to assign shares in a way that minimizes the streaming cost.
%The goal is to assign shares so as to minimize $\delta(G)$.

\begin{problem}\label{def:continuous_problem} Given a weighted directed acyclic graph $G$ and a capacity $c>0$, find a mapping $x: V \rightarrow \mathbb{R}^+$ that minimizes $\delta(G)$.
\end{problem}

Obviously, it must hold that $\sum_{v \in V} x(v) = c$ in an optimal allocation, i.e., the entire capacity is assigned. The motivation for studying this problem is that the optimal solution of Problem~\ref{def:continuous_problem} is a lower bound on the streaming cost in the discrete model.

\begin{theorem} For all graphs $G$ and $c \in \mathbb{N}$ it holds that $\delta(G) \le d(G)$, where $c$ is the capacity in the continuous case and the number of resources in the discrete case.
\end{theorem}
%\begin{comment}
\begin{proof}
Consider the optimal solution $r$ of Problem~\ref{def:general_problem}. Set $x(v) := 1/n(r(v))$ for all $v\in V$. It holds that
$$\sum_{v\in V} \frac{1}{n(r(v))} = \sum_{R \in \mathcal{R}}\hspace{2pt} \sum_{v\in V: r(v)=R} \frac{1}{n(r(v))}= \sum_{R \in \mathcal{R}} \frac{n(R)}{n(R)}= c$$
due to the fact that there are $n(r(v))=n(R)$ tasks assigned to resource $R$ by definition. As the sum of shares is $c$, it is a valid solution for Problem~\ref{def:continuous_problem}. Thus, the optimal continuous cost $\delta(G)$ can only be lower or equal.
\end{proof}
%\end{comment}

The two problems are indeed strongly related. If a resource is shared among $k$ tasks, the processing cost in the discrete model is $d(v) = k\cdot w(v)$ for each such task $v$. In other words, each task gets a \emph{share} of $1/k$ of the resource, which in the continuous model corresponds to a processing cost of $\delta(v) = w(v) /x(v) = k\cdot w(v)$, i.e., a share $x(v)$ can be interpreted as the fraction of a resource dedicated to $v$.
Of course, the continuous model does not truly have a notion of a ``resource'' as the capacity $c$ can be split up arbitrarily. Moreover, it is admissible to assign a share greater than 1 to a task in the continuous model, which would mean that a task is assigned to more than one resource.
Nevertheless, the relation between the problems can be exploited.
First, we formulate and analyze an algorithm, $\mathcal{A}^{cont}$, which solves Problem~\ref{def:continuous_problem} for SPD graphs, then we present  an algorithm $\mathcal{A}^{disc}$ for the discrete case, which uses $\mathcal{A}^{cont}$ as a subroutine.

Algorithm  $\mathcal{A}^{cont}$ takes $z_0$, i.e., the root of the SPD tree corresponding to graph $G$, and $c$ as input parameters and computes the optimal mapping $x: V \rightarrow \mathbb{R}^+$. To this end, it first calls procedure \textit{computeWeights} with the parameter $z_0$, which computes weights for each node in the tree. The weight of a node $z$ corresponds to the optimal streaming cost of the subtree rooted at $z$.
Next, procedure \textit{computeMapping} is called with parameters $z_0$ and $c$, which derives the optimal mapping $x$ based on the weights computed in the previous step.

Procedure \textit{computeWeights} (see Algorithm~\ref{algo:computeWeights}) recursively computes the weights of all children of a node $z$. The weight $w(z)$ of a leaf $z$ equals the weight $w(v)$ of the corresponding graph node $v\in V$.
% $z = (\bot,\{v\})$ equals $w(v)$, i.e., the weight of the node. 
The weight of an internal node $z$ is computed from the weights of its children:   $w(z)$ is set to $(\sum_{z_i \in \mathcal{C}(z)} \sqrt{w(z_i)})^2$ if $op = s$ and $\sum_{z_i \in \mathcal{C}(z)} w(z_i)$ if $op = p$.

\begin{algorithm}[t]
\small
   \caption{computeWeights($z$)}
   	\label{algo:computeWeights}
	\begin{algorithmic}[1]
	%\IF{$\mathcal{C}(z) \ne (\bot,\cdot)$}
	  \FOR{$z_i \in \mathcal{C}(z)$}
            \STATE computeWeights($z_i$)
            \IF{$z = (s,\cdot)$}
	      \STATE $w(z) := \left(\sum_{z_i \in \mathcal{C}(z)} \sqrt{w(z_i)}\right)^2$
            \ELSE
              \STATE $w(z) := \sum_{z_i \in \mathcal{C}(z)} w(z_i)$
            \ENDIF
          \ENDFOR
        %\ENDIF
     \end{algorithmic}
  \end{algorithm}

Procedure \textit{computeMapping} (see Algorithm~\ref{algo:computeMapping}) traverses the SPD tree in a top-down fashion and maps the capacity to nodes. %, unlike procedure \textit{computeWeights}. 
It recursively computes the partitioning of the given capacity, which is $c$ at the root $z_0$, among all children. As in procedure \textit{computeMapping}, the partitioning is different for serial and parallel compositions. Each child gets a share relative to its contribution to the sum of weights for parallel compositions, whereas the relative contribution with respect to the square roots of the weights is used for serial compositions.  When the recursion arrives at a leaf $z$ with the call  \textit{computeMapping}$(z,x)$ it receives the share $x$, which implies that $x(v) = x$ for the corresponding graph node.

In order to prove that the computed mapping is optimal, we must show that the mapping rules lead to an optimal solution, under the assumption that the computed weight of each child is the optimal streaming cost of the corresponding subtree.

\begin{lemma}\label{lemma:continuous}
Let $z = (op,\{z_1,\ldots,z_k\})$ and $w(z_i)$ be the optimal streaming cost of $z_i$. For all $i \in \{1,\ldots,k\}$, the capacity $c$ is partitioned optimally by setting
\begin{numcases}{ x(z_i) :=}
\frac{c\cdot\sqrt{w(z_i)}}{\sum_{j=1}^k \sqrt{w(z_j)}} & if $op=s$\label{eq:serial}\\
\frac{c\cdot w(z_i)}{\sum_{j=1}^k w(z_j)}  & if $op=p$.\label{eq:parallel}
\end{numcases}
\end{lemma}
\begin{proof}
 Let $x_i := x(z_i)$. The multivariate function that describes the streaming cost of $z$ is
\begin{displaymath}
w_s(x_1,\ldots,x_{k-1}) := \left(\sum_{i=1}^{k-1} \frac{w(z_i)}{x_i}\right) + \frac{w(z_k)}{c-\sum_{j=1}^{k-1} x_j}.
\end{displaymath}
The minimum of this function is attained when $\frac{\partial w_s}{\partial x_i} = 0$ for all $i = 1,\ldots,k-1$, which implies that
\begin{displaymath}
\forall i \in \{1,\ldots,k-1\}: \frac{w(z_i)}{x_i^2} = \frac{w(z_k)}{\left(c-\sum_{j=1}^{k-1} x_j\right)^2}.
\end{displaymath}
Solving this equation for $x_i$ yields $x_i = \frac{c\cdot\sqrt{w(z_i)}}{\sum_{j=1}^k \sqrt{w(z_j)}}$ for all $i\in \{1,\ldots,k\}$ as claimed.

The same pattern can be used to derive the optimal partitioning for parallel compositions.
As before, let $x_i := x(z_i)$, and the multivariate function for the streaming cost of $z$ is
\begin{displaymath}
w_p(x_1,\ldots,x_k) := \max\left(\frac{w(z_1)}{x_1},\ldots,\frac{w(z_k)}{c - \sum_{i=1}^{k-1} x_i}\right).
\end{displaymath}
This function is minimized if each term is equal, which is the case if $c \cdot w(z_i)/x_i = \sum_{j=1}^k w(z_j)$ as claimed.
\end{proof}

 \begin{algorithm}[t]
\small
\caption{computeMapping($z,c$)}
   \label{algo:computeMapping}
   \begin{center}
    \begin{algorithmic}[1]
	\IF{$z$ is leaf}
            \STATE $x(z) := c$
        \ELSE
	 \FOR{$z_i \in \mathcal{C}(z)$}
            \IF{$z = (s,\cdot)$}
              \STATE $W := \sum_{z_i \in \mathcal{C}(z)} \sqrt{w(z_i)}$
	      \STATE computeMapping($z_i,c \cdot \sqrt{w(z_i)}/W$)
            \ELSE
              \STATE $W := \sum_{z_i \in \mathcal{C}(z)} w(z_i)$
              \STATE computeMapping($z_i,c \cdot w(z_i)/W$)
            \ENDIF
          \ENDFOR
        \ENDIF
%\ENDIF
\vspace{-5pt}
     \end{algorithmic}
   \end{center}
 \end{algorithm}

An important observation is that the optimal partitioning scales linearly with $c$ in both cases, i.e., the relative distribution among the constituent parts is independent of $c$.

\begin{fact}\label{fact:linear}
The optimal partitioning for both serial and parallel compositions scales linearly with the capacity $c$.
\end{fact}

This fact is important as it implies that the optimal solution can be built recursively as long as the optimal costs of all subtrees are known, which is exactly what Algorithm $\mathcal{A}^{cont}$ does. The following theorem states the main result.

\begin{theorem}
Algorithm $\mathcal{A}^{cont}$ computes an optimal mapping \mbox{$x: V \rightarrow \mathbb{R}^+$} for any SPD graph $G$ and capacity $c > 0$.
\end{theorem}
\begin{proof}
Lemma~\ref{lemma:continuous} and Fact~\ref{fact:linear} show that procedure \textit{computeMapping} optimally partitions the capacity in a recursive manner under the assumption that all weights correspond to the minimal streaming cost of the respective subtree. It remains to prove that procedure \textit{computeWeights} indeed computes the optimal weights.

We can use an inductive argument to prove this. As the capacity $c$ merely scales the optimal solution linearly, we can ignore it when computing the weights by setting $c := 1$.
Consider the base case of a serial or parallel composition of leaves, i.e., real nodes in $V$. Let $z$ denote the root of the SPD tree of this subgraph. If it is a serial composition, Lemma~\ref{lemma:continuous} reveals that the optimal streaming cost is 
\begin{align*}
w(z) &= \sum_{z_i\in \mathcal{C}(z)} \frac{w(z_i)}{x(z_i)}    \stackrel{\eqref{eq:serial}}{=} \sum_{z_i\in \mathcal{C}(z)} \frac{w(z_i)}{\sqrt{w(z_i)}}\left(\sum_{z_j\in \mathcal{C}(z)} \sqrt{w(z_j)}\right) \\
 &= \left(\sum_{z_i\in \mathcal{C}(z)} \sqrt{w(z_i)}\right)^2.
\end{align*}

Similarly, we can use Lemma~\ref{lemma:continuous} to get the optimal streaming cost for a parallel composition, which is 
\begin{align*}
w(z) & = \max_{z_i \in \mathcal{C}(z)} \frac{w(z_i)}{x(z_i)} 
 \stackrel{\eqref{eq:parallel}}{=} \max_{z_i \in \mathcal{C}(z)} \frac{w(z_i)}{w(z_i)}\sum_{z_i \in \mathcal{C}(z)} w(z_i) \\
 &= \sum_{z_j\in \mathcal{C}(z)} w(z_i).
\end{align*}

Thus, procedure \textit{computeWeights} computes the optimal cost, i.e., weight, in both cases and by induction, all weights are computed optimally for the entire graph.
\end{proof}

\subsection{Discrete Algorithm}
We use the optimal continuous algorithm $\mathcal{A}^{cont}$ to devise an algorithm, $\mathcal{A}^{disc}$, for (the discrete) Problem~\ref{def:general_problem}.
%As mentioned earlier, our algorithm for (the discrete) Problem~\ref{def:general_problem}, which we call algorithm $\mathcal{A}^{disc}$, takes advantage of $\mathcal{A}^{cont}$. 
The main idea is to use the optimal continuous shares as an indicator for the number of tasks that should be mapped to individual resources. Algorithm $\mathcal{A}^{disc}$ is a greedy algorithm that allocates tasks to resources starting with the tasks with the largest shares, i.e., the tasks mapped to resources that must not be shared with many other tasks. The algorithm must overcome two issues: First, it is not possible to allocate tasks greedily in such a way that $n(r(v))$ is proportional to $1/x(v)$ for all $v \in V$.
The second issue is that some shares may be larger than 1. An illustrative pathological example is the case where one task $v$ has an exorbitantly large weight, resulting in a share of $x(v) \lesssim c$. If $d(G) \in \BigO(d(v)) = \BigO(w(v)/c)$, the best possible discrete solution is at least a factor of $c$ worse as $d(v) \ge w(v)$ for all $v \in V$ and all allocations.

We will now present Algorithm $\mathcal{A}^{disc}$, which is given in Algorithm~\ref{algo:discrete}, and show how it overcomes the aforementioned issues.
After computing the optimal continuous shares, the largest share $x(v)$ for some task $v \in V$ is rounded down and fixed to $1$ if it exceeds this threshold. It is fixed in the sense that task $v$ is removed from the optimization problem and replaced with the constant $d(v)$. Subsequently, algorithm $\mathcal{A}^{cont}$ is executed again with this added constraint. If the largest share still exceeds $1$, the same steps are carried out until all shares are upper bounded by $1$ (Lines~1-9).
These modified shares are then used to allocate the tasks to the $c$ resources as follows. The shares are first sorted in decreasing order, resulting in shares $\bar{x}_1 \ge \ldots \ge \bar{x}_n$. The algorithm then performs a single pass over the sorted shares, starting at the largest value. The tasks with the $\lceil\frac{2n^{2/c}}{\bar{x}_1}\rceil$ largest shares are assigned to resource $R_1$. The share $\bar{x}_i$ at index $i = \lceil\frac{2n^{2/c}}{\bar{x}_1}\rceil+1$ determines how many resources are assigned to $R_2$, i.e., $\lceil\frac{2n^{2/c}}{\bar{x}_i}\rceil$ many. This process is repeated until all nodes are assigned to resources (Lines~9-15).

\begin{algorithm}[t]
\small
\caption{Algorithm $\mathcal{A}^{disc}$ takes $z_0$ and $c$ as parameters and computes a mapping $r:V\rightarrow \mathcal{R}$}
   \label{algo:discrete}
%\begin{multicols}{2}
   \begin{center}
    \begin{algorithmic}[1]
	\REPEAT
	  \STATE \textbf{execute} Algorithm $\mathcal{A}^{cont}$
	  \STATE $\mathcal{S} := \{v \in V \;|\; x(v) > 1\}$
	  \IF{$\mathcal{S} \neq \emptyset$}
	    \STATE $v_{max} := \arg\max_{v \in \mathcal{S}} x(v)$
	    \STATE $x(v_{max}) := 1$
	    \STATE \textbf{remove} $v_{max}$
	  \ENDIF
	\UNTIL{$\mathcal{S} = \emptyset$}
        \STATE $\{\bar{x}_1,\ldots,\bar{x}_n\}$ $:=$ sort-decreasing($\{x(v_1),\ldots,x(v_n)\}$)
        \STATE $index := 1$; $k := 1$
        \WHILE{$index \le n$}
	  \STATE $size := \left\lceil\frac{2n^{2/c}}{\bar{x}_{index}}\right\rceil$
	  \FOR{$i=index,\ldots,\min\{index+size-1,n\}$}
	    \STATE $r(v) := R_k$ \textbf{where} $x(v) = \bar{x}_i$
	  \ENDFOR
	  \STATE $index := index+size$; $k := k+1$
	\ENDWHILE
\vspace{-5pt}
     \end{algorithmic}
   \end{center}
%\end{multicols}
 \end{algorithm}

%The following lemma allows us to assess the quality of the allocation by proving that the modified shares are still optimal for the case when $x(v) \le 1$ must hold for all $v \in V$.

Lemma~\ref{lma:optConstrained} shows that Algorithm~$\mathcal{A}^{disc}$ modifies shares $x$ in a way that preserves optimality for the case when shares cannot exceed the capacity of resources. Subsequently, we state the main result in Theorem~\ref{thm:approximation}.

\begin{lemma}\label{lma:optConstrained}
Lines~1-8 in Algorithm~\ref{algo:discrete} compute optimal shares for any SPD graph G and capacity $c > 0$ subject to the constraint that shares must not exceed $1$.
\end{lemma}
\begin{proof}
Consider task $v$ with the largest share $x(v) > 1$. Assume for the sake of contradiction that the optimal constrained share should be $x(v) < 1$. % in adding the constraint results in an optimal share of $x(v) < 1$.
Let $z$ be the parent node of $z_i=v$ in the SPD tree.
%We will make a case distinction. 
If the capacity is not distributed according to Equation~\eqref{eq:serial} (Equation~\eqref{eq:parallel}) in a serial (parallel) composition, the distribution can be changed locally, i.e., among $z$ and $\mathcal{C}(c)$, to the optimal distribution, which reduces $w(z)$ and, inductively, also $w(z_0) = \delta(G)$, contradicting optimality of the shares. 
Otherwise, the share distribution among $z$ and $\mathcal{C}(c)$ is optimal, but $x(z)$ is smaller than the $x(z)$ calculated by $\mathcal{A}^{disc}$. Tracing the cause of the lower share towards the root, we find that either the capacity was not distributed optimally or, again, a share that is too small was assigned at this level. If we arrive at $z_0$, and the capacity is distributed optimally, it must be the case that less than the full capacity $c$ was assigned, which cannot be optimal.  
This argument establishes that $x(v)$ must equal $1$. As the optimal shares are recomputed under this constraint, the same argument can be used inductively for the next largest share exceeding $1$, which proves the claim.
\end{proof}

%We are now in the position to prove the main result.

\begin{theorem}\label{thm:approximation}
Algorithm $\mathcal{A}^{disc}$ computes an $\BigO(n^{\BigO(1/c)})$-approximation for any computationally constrained SPD graph $G$. 
\end{theorem}
\begin{proof}
The allocation strategy of $\mathcal{A}^{disc}$ ensures that $n(r(v)) \le \lceil\frac{2n^{2/c}}{x(v)}\rceil \le \frac{2n^{2/c}}{x(v)}+1$ for all $v$ that are allocated \emph{first} to a resource $r(v)$. Since $x(v) \ge x(w)$ for any other task $w$ for which $r(v) = r(w)$, the inequality generally holds for all tasks. 
Therefore, we have for all $v \in V$ that $d(v) \in \BigO(n^{\BigO(1/c)}\delta(v))$, which implies that
%\begin{displaymath}
$d(G) \in \BigO(\hat{d}(G)) \in \BigO\left(n^{\BigO(1/c)}\delta(G)\right)$
%\end{displaymath}
for any computationally constrained SPD graph $G$.
It remains to show that all tasks are allocated to resources. Let
%\begin{displaymath}
$\rho(R_i) := \sum_{v \in V : r(v) = R_i} x(v).$
%\end{displaymath}
It suffices to show that $\sum_{i=1}^c \rho(R_i) \ge c$.
We define $\bar{x}^{(i)}$ as the largest capacity assigned to resource $R_i$, i.e., $\bar{x}^{(1)} := \bar{x}_1$ and  $\bar{x}^{(2)} := \bar{x}_j$, where $j = \left\lceil\frac{2n^{2/c}}{\bar{x}_1}\right\rceil+1$ and so on. 
Further let $\bar{x}^{(c+1)}:= \bar{x}_n$. 
Since the shares are ordered, we know that for all $i \in \{1,\ldots,c\}$, it holds that
\begin{IEEEeqnarray}{lCr}
\rho(R_i) &\ge & \bar{x}^{(i)} + \left(\left\lceil\frac{2 n^{2/c}}{\bar{x}^{(i)}}\right\rceil-1\right)\bar{x}^{(i+1)} \nonumber \\
&\ge& \bar{x}^{(i)} + \left(\frac{2 n^{2/c}}{\bar{x}^{(i)}}-1\right)\bar{x}^{(i+1)}.\label{eq:rho_i}
\end{IEEEeqnarray}
For the sum of all $\rho(R_i)$ we get that
%\begin{align*}
$$\sum_{i=1}^c \rho(R_i) \stackrel{\eqref{eq:rho_i}}{\ge} \bar{x}^{(1)} - \bar{x}_n +\sum_{i=1}^c 2n^{2/c}\frac{\bar{x}^{(i+1)}}{\bar{x}^{(i)}}
\ge \sum_{i=1}^c 2n^{2/c}\frac{\bar{x}^{(i+1)}}{\bar{x}^{(i)}}.$$
%\end{align*}
If at least half of the terms is at least $2$, then the total sum is at least $c$, which means that all tasks can be allocated to the resources. For the sake of contradiction, assume that more than half of the terms are smaller than $2$. For each such term $2n^{2/c}\frac{\bar{x}^{(i+1)}}{\bar{x}^{(i)}} < 2$, it holds that
\begin{equation}
\bar{x}^{(i+1)} < \frac{\bar{x}^{(i)}}{n^{2/c}}.\label{eq:decrease}
\end{equation}
After the first $c/2$ cases where Inequality~\eqref{eq:decrease} holds, we get for the corresponding index $j$ that
\begin{displaymath}
\bar{x}^{(j)} \stackrel{\eqref{eq:decrease}}{<} \frac{\bar{x}^{(1)}}{\left(n^{2/c}\right)^{c/2}} \le \frac{1}{n}.
\end{displaymath}
Thus, all subsequent shares are so small that the respective tasks can be allocated to a single resource, and the processing cost for those tasks is not larger than in the continuous case, implying that more than half of the terms cannot satisfy Inequality~\eqref{eq:decrease}, which concludes the proof.
\end{proof}

Note that we used the constant $2$ twice in Line~12 of Algorithm~\ref{algo:discrete} for ease of exposition. It is straightforward to compute optimal constants for a given $n$ and $c$.

%\subsection{Generalizations \& Extentions}

\subsection{Non-Computationally Constrained Problems} \label{app:greedy}
As shown, %in Section~\ref{sec:algorithm}, 
the streaming cost for task allocation remains bounded with algorithm $\mathcal{A}^{disc}$ if the streaming graph is computationally constrained. If streaming graphs are not computationally constrained, the question arises how to deal with (arbitrarily) large transfer costs.
Specifically, we discuss whether we can build upon the techniques used by $\mathcal{A}^{disc}$: As we took a greedy approach to derive a discrete solution from the continuous solution for computationally constrained streaming graphs, we look into the difficulty of applying greedy strategies to handle large transfer costs.

The blueprint of our greedy strategies is the following: After computing the optimal continuous solution,
the algorithm traverses all edges from heaviest to lightest. For an edge $e = (u,v)$, it adds a constraint that enforces tasks $u$ and $v$ to be collocated, i.e., $r(u) = r(v)$; then, Algorithm $\mathcal{A}^{disc}$ is executed with the newly added constraint to discover a new allocation $r$. Depending on the quality of $r$, the constraint is retained or dismissed. After the traversal, the allocation computed by $\mathcal{A}^{disc}$ under the retained constraints is the solution.
The considered greedy strategies differ in the retention policy for constraints.
We now show that multiple intuitive strategies fail to achieve a better approximation ratio than the trivial $\BigO(n)$ bound.

Arguably the most straightforward strategy is to retain a constraint if adding it results in a reduced streaming cost $d(G)$.
This strategy may already fail if there are two paths $P_1$ and $P_2$ for which $d(P_1) = d(P_2) = d(G)$ and there is an edge with an arbitrarily large transfer cost on each path: if both large edges are not covered, i.e., adjacent tasks are not collocated, then each individual constraint may reduce the cost of the respective path, but not $d(G)$; applying both constraints together, however, would result in a reduction of $d(G)$. As the reduction is not bounded, neither is the approximation ratio.

The deficiency of the strategy above can be overcome by slightly changing the rule to always retain a constraint \emph{unless} it increases $d(G)$. However, also this strategy $\mathcal{S}$ fails in that the approximation ratio may grow linearly with the number of tasks, even for a large number of resources.

\begin{figure}[t]
\center
 \includegraphics[width=.85\columnwidth]{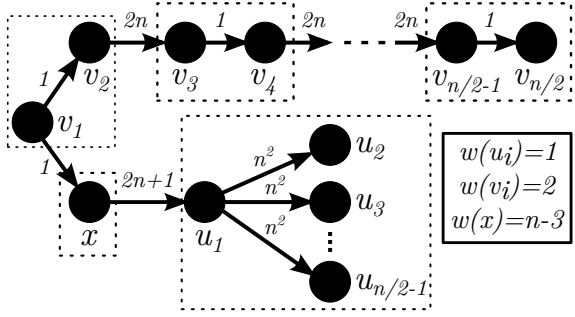}
 \caption{SPD graph where the approximation ratio of the greedy strategy $\mathcal{S}$ is in $\Omega(n)$.}
 \label{fig:greedy-worst-case}
\end{figure}

\begin{theorem}
The approximation ratio of the streaming cost of strategy $\mathcal{S}$ is in $\Omega(n)$ even for $c \in \Theta(n)$.
\end{theorem}
\begin{proof}
Consider the streaming graph depicted in Figure~\ref{fig:greedy-worst-case}. Let the number of resources be $c := \frac{n}{4}+2$.
Node and edge weights are as stated in the figure. The dashed boxes illustrate the optimal task allocation if transfer costs are disregarded: It holds that $\hat{d}(\{v_1,v_2,\ldots,v_{n/2}\}) = 2n$ and $\hat{d}(\{v_1,x,u_1,u_i\}) = 2n-1$ for all $i = 2,\ldots,\frac{n}{2}-1$, and therefore $\hat{d}(G) = 2n$.
If we consider transfer costs as well, strategy $\mathcal{S}$ will ensure that all tasks $u_i$ are collocated.
At this stage, it holds for all $ i = 2,\ldots,\frac{n}{2}-1$ that
\begin{IEEEeqnarray*}{rCl}
d(\{v_1,v_2,\ldots,v_{n/2}\}) = 2n + 2n\cdot(\frac{n}{4}-1) &=& \frac{n^2}{2} \text{,~~~~and} \\ 
d(\{v_1,x,u_1,u_i\}) = 2n - 1 + (2n+2) &=& 4n+1.
\end{IEEEeqnarray*}
Thus, it holds that $d(G) = \frac{n^2}{2}$.
Strategy $\mathcal{S}$ will retain the constraint for the next heaviest edge $(x,u_1)$ because the new cost on path $P = \{v_1,x,u_1,u_i\}$ is then
\begin{displaymath}
 d(P) = 5 + \frac{n^2}{2} - \frac{n}{2} \le \frac{n^2}{2}
\end{displaymath}
for all $i = 2,\ldots,\frac{n}{2}-1$ and $n\ge 10$. As $\mathcal{S}$ will keep these constraints, it holds that $d(G) \in \Omega(n^2)$ also at termination.
However, the optimal solution does not collocate tasks $x$ and $u_1$, instead it reassigns all pairs on the path $\{v_1,v_2,\ldots,v_{n/2}\}$ so that the endpoints of edges with transfer costs of $2n$ are collocated, resulting in a streaming cost of $$d(\{v_1,v_2,\ldots,v_{n/2}\}) = 2n + \frac{n}{4} \in \BigO(n),$$ which proves the claimed bound.
\end{proof}

Various other strategies, such as % retaining all constraints from the heaviest edge down to any edge $e$, or
keeping the constraint for each edge $e$ as long as the streaming cost of any path including $e$ does not increase, also fail to achieve a better bound.
While these results do not preclude the existence of a simple strategy that guarantees a good approximation ratio, they suggest that a different approach may be required to cope with large edge weights in the general case.

%% file: relatedwork.tex
% !TeX root = distributedStreamProcessing.tex
\section{Related Work}
\label{sec:relatedwork}

There is a large body of work on (distributed) stream processing, covering a broad variety of topics.
The architectural challenges concerning \emph{scalability}, \emph{load management}, \emph{high availability}, and \emph{federated operation} have been addressed (see, e.g., \cite{cherniack03}), as well as the requirements and algorithmic challenges in stream processing in general (e.g., \cite{babcock02}).
As a result, several general-purpose stream processing platforms have been proposed~\cite{arasu03,chandrasekaran03,abadi05,abadi03,carney02},  enabling the processing of continuous streams from many sources by providing primitive operators, which are building blocks in the form of functions, to build up complex stream processing topologies.

Apart from scalability and availability, \emph{adaptive control} is a key requirement to absorb variable and bursty workloads~\cite{amini06}.
Xing et al.\ propose a correlation-based algorithm that strives to minimize both average load and load variance on the resources to protect against bursty inputs~\cite{xing05}. The basic idea is to measure the correlation coefficient over time and to collocate tasks with small coefficients. Unlike our definition of processing cost, \emph{load} is defined as the sum of the costs of tasks, where the cost of a task is the arrival rate multiplied by the processing time. While the load distribution problem is NP-hard, it is shown that a greedy algorithm yields a fairly good distribution in practice. It is further illustrated that the streaming model is fundamentally different from other parallel processing models. 
\emph{Load shedding} is another approach to coping with excessive load. Tatbul et al.\ model the distributed load shedding problem as a linear optimization problem subject to preserving low-latency processing and minimizing quality degradation~\cite{tatbul07}.
By contrast, our work considers only static allocations at invariable input rates. However, since our algorithm is efficient, it can cope with changing environments through periodic re-execution. Similarly, Chatzistergiou et al.~\cite{chatzistergiou14} suggest to use greedy task allocation for fast reallocation in dynamic environments. They experimentally evaluate their algorithm, which is tailored to problems consisting of groups of similar-weight tasks.   
Although their problem definition includes arbitrary DAGs, the studied streaming networks are typically SPD graphs with restricted weights. Their model is different in that processing cost of a task is independent of the allocation.

Another key goal is to make stream processing \emph{easily accessible}. To this end, platforms have been built that support efficient development of applications for processing continuous unbounded streams of data by exposing a set of simple programming interfaces. Examples of such platforms are S4~\cite{neumeyer10} and Storm\footnote{See http://storm.apache.org/.}. Both enable a programmatic specification of a level of parallelism for a processing element (PE) prototype, which determines the number of parallel instances to be executed---a feature that directly corresponds to our definition of parallel composition.
System S takes this approach one step further in that it provides its own language (SPADE) for the composition of parallel data-flow graphs~\cite{gedik08}. In the context of System S, it has also been studied how to coalesce basic operators into PEs and how to distribute them onto available hosts~\cite{khandekar09}. Their work differs from ours as the PEs are given in our model and we consider different cost functions.
Additionally, the automatic composition of System S workflows, i.e., streaming topologies, has also been investigated~\cite{riabov06}.

While there is no theoretical work on a task allocation model similar to ours, the problem of allocating resources and placing operators in stream processing has received much attention. Wolf et al.\ propose a scheduler that shifts the allocation dynamically in the face of changes in resource availability and job arrival and departures in order to optimize the weighted average of the allocation quality~\cite{wolf08}. A key difference to our model is that \emph{fractional assignments} of tasks to resources are possible in their model.
Xing et al.\ have introduced the notion of a \emph{resilient operator (i.e., task) placement plan}, which is resilient in the sense that it can withstand the largest set of input rate combinations~\cite{xing06}. An algorithm for their model---where load functions can be expressed as linear constraint sets---is shown to improve resilience experimentally. Another approach to operator placement makes use of a layer between the stream processing system and the physical network that determines the placement in a virtual cost space and then maps the cost-space coordinates to physical resources~\cite{pietzuch06}.
Finally, Mattheis et al.\ adapt \emph{work stealing} strategies for stream processing and give bounds on the maximum latency for certain stealing strategies~\cite{mattheis12}.

Task allocation for stream processing is also related to \emph{precedence-constrained scheduling}, 
where a set of jobs with precedence constraints has to be scheduled on $m$ processors so as to minimize the makespan or the average job completion time. These problems, which are generally NP-hard, have been extensively studied already in the 1970s (see~\cite{hartmann10} and references therein).
The key difference to stream processing is that each job is executed only once; thus, as soon as a job is completed, no more resources need to be invested in this job. In stream processing, an allocated task is continuously executed on this resource (in parallel to collocated tasks).

There has been considerable interest in \emph{series-parallel graphs} for many years due to their versatility and the fact that many NP-hard problems are solvable in polynomial time on these graphs. Examples for such problems are the minimum vertex cover problem, the maximum matching problem, and the maximum disjoint triangle problem~\cite{takamizawa82}. 
Most related to our work is the topic of scheduling jobs subject to precedence constraints in the form of a series-parallel graph. It has been shown how to minimize the makespan for deteriorating jobs, for which the processing time increases with the start delay, in polynomial time~\cite{wang08}.
The work most similar to ours studies scheduling of task graphs on two identical processors, where tasks have unit execution times and unit communication delays, and communication is free between collocated tasks. While this problem is NP-hard for general graphs~\cite{rayward87}, Finta et al.\ present an algorithm that computes an optimal schedule for a class of series-parallel graphs in quadratic time~\cite{finta96}.
Despite the similarities, their model is quite different in that there are no precedence relations between the tasks in our model as all allocated tasks must be executed in parallel, which necessitates a completely different approach.

%% file: conclusion.tex
\section{Conclusion}
\label{sec:conclusion}
We have introduced a theoretical model and a task allocation problem where tasks of a streaming topology must be allocated to a fixed set of physical resources for continuous processing.
As the problem is NP-hard, we have focused on approximation algorithms and presented an algorithm whose cost is only a small factor larger than in the optimal case under certain assumptions. While the algorithm solves the problem for an important case, there is much left to explore. In particular, the resources may not have uniform capacities and bandwidths, which makes it harder to find an optimal allocation. An interesting open question is also what guarantees on the approximation quality can be given in polynomial time for arbitrary directed acyclic streaming topologies.

%% file: distributedStreamProcessing.bbl
\begin{thebibliography}{10}
\providecommand{\url}[1]{#1}
\csname url@samestyle\endcsname
\providecommand{\newblock}{\relax}
\providecommand{\bibinfo}[2]{#2}
\providecommand{\BIBentrySTDinterwordspacing}{\spaceskip=0pt\relax}
\providecommand{\BIBentryALTinterwordstretchfactor}{4}
\providecommand{\BIBentryALTinterwordspacing}{\spaceskip=\fontdimen2\font plus
\BIBentryALTinterwordstretchfactor\fontdimen3\font minus
  \fontdimen4\font\relax}
\providecommand{\BIBforeignlanguage}[2]{{%
\expandafter\ifx\csname l@#1\endcsname\relax
\typeout{** WARNING: IEEEtran.bst: No hyphenation pattern has been}%
\typeout{** loaded for the language `#1'. Using the pattern for}%
\typeout{** the default language instead.}%
\else
\language=\csname l@#1\endcsname
\fi
#2}}
\providecommand{\BIBdecl}{\relax}
\BIBdecl

\bibitem{akidau13}
T.~Akidau, A.~Balikov, K.~Bekiro{\u{g}}lu, S.~Chernyak, J.~Haberman, R.~Lax,
  S.~McVeety, D.~Mills \emph{et~al.}, ``{MillWheel: Fault-Tolerant Stream
  Processing at Internet Scale},'' \emph{Proc. VLDB Endowment}, vol.~6, no.~11,
  pp. 1033--1044, 2013.

\bibitem{arasu03}
A.~Arasu, B.~Babcock, S.~Babu, M.~Datar, K.~Ito, I.~Nishizawa, J.~Rosenstein,
  and J.~Widom, ``{STREAM: the Stanford Stream Data Manager (Demonstration
  Description)},'' in \emph{Proc. ACM International Conference on Management of
  Data}, 2003, pp. 665--665.

\bibitem{chandrasekaran03}
S.~Chandrasekaran, O.~Cooper, A.~Deshpande, M.~J. Franklin, J.~M. Hellerstein,
  W.~Hong, S.~Krishnamurthy, S.~R. Madden \emph{et~al.}, ``{TelegraphCQ:
  Continuous Dataflow Processing},'' in \emph{Proc. 2003 ACM International
  Conference on Management of Data}, 2003, pp. 668--668.

\bibitem{gedik08}
B.~Gedik, H.~Andrade, K.-L. Wu, P.~S. Yu, and M.~Doo, ``{SPADE: The System S
  Declarative Stream Processing Engine},'' in \emph{Proc. ACM International
  Conference on Management of Data}, 2008, pp. 1123--1134.

\bibitem{neumeyer10}
L.~Neumeyer, B.~Robbins, A.~Nair, and A.~Kesari, ``{S4: Distributed Stream
  Computing Platform},'' in \emph{Proc. IEEE International Conference on Data
  Mining Workshops (ICDMW)}, 2010, pp. 170--177.

\bibitem{dean04}
J.~Dean and S.~Ghemawat, ``{MapReduce: Simplified Data Processing on Large
  Clusters},'' in \emph{Proc. 6th Conference on Symposium on Opearting Systems
  Design \& Implementation (OSDI)}, 2004, pp. 137--150.

\bibitem{takamizawa82}
K.~Takamizawa, T.~Nishizeki, and N.~Saito, ``{Linear-time Computability of
  Combinatorial Problems on Series-parallel Graphs},'' \emph{Journal of the ACM
  (JACM)}, vol.~29, no.~3, pp. 623--641, 1982.

\bibitem{mertens06}
S.~Mertens, ``{The Easiest Hard Problem: Number Partitioning},'' in
  \emph{{Computational Complexity and Statistical Physics}}, A.~Percus,
  G.~Istrate, and C.~Moore, Eds.\hskip 1em plus 0.5em minus 0.4em\relax Oxford
  University Press, 2006, pp. 125--139.

\bibitem{cherniack03}
M.~Cherniack, H.~Balakrishnan, M.~Balazinska, D.~Carney, U.~Cetintemel,
  Y.~Xing, and S.~B. Zdonik, ``{Scalable Distributed Stream Processing},'' in
  \emph{Proc. 1st Biennial Conference on Innovative Data Systems Research
  (CIDR)}, vol.~3, 2003, pp. 257--268.

\bibitem{babcock02}
B.~Babcock, S.~Babu, M.~Datar, R.~Motwani, and J.~Widom, ``{Models and Issues
  in Data Stream Systems},'' in \emph{Proc. 21st ACM Symposium on Principles of
  Database Systems (PODS)}, 2002, pp. 1--16.

\bibitem{abadi05}
D.~J. Abadi, Y.~Ahmad, M.~Balazinska, U.~{\c{C}}etintemel, M.~Cherniack, J.-H.
  Hwang, W.~Lindner, A.~Maskey \emph{et~al.}, ``{The Design of the Borealis
  Stream Processing Engine},'' in \emph{Proc. 2nd Biennial Conference on
  Innovative Data Systems Research (CIDR)}, vol.~5, 2005, pp. 277--289.

\bibitem{abadi03}
D.~J. Abadi, D.~Carney, U.~{\c{C}}etintemel, M.~Cherniack, C.~Convey, S.~Lee,
  M.~Stonebraker, N.~Tatbul, and S.~Zdonik, ``{Aurora: A New Model and
  Architecture for Data Stream Management},'' \emph{The International Journal
  on Very Large Data Bases}, vol.~12, no.~2, pp. 120--139, 2003.

\bibitem{carney02}
D.~Carney, U.~{\c{C}}etintemel, M.~Cherniack, C.~Convey, S.~Lee, G.~Seidman,
  M.~Stonebraker, N.~Tatbul, and S.~Zdonik, ``{Monitoring Streams: A New Class
  of Data Management Applications},'' in \emph{Proc. 28th International
  Conference on Very Large Data Bases (VLDB)}, 2002, pp. 215--226.

\bibitem{amini06}
L.~Amini, N.~Jain, A.~Sehgal, J.~Silber, and O.~Verscheure, ``{Adaptive of
  Control Extreme-scale Stream Processing Systems},'' in \emph{Proc. 26th IEEE
  International Conference on Distributed Computing Systems (ICDCS)}, 2006, pp.
  71--71.

\bibitem{xing05}
Y.~Xing, S.~Zdonik, and J.-H. Hwang, ``{Dynamic Load Distribution in the
  Borealis Stream Processor},'' in \emph{Proc. 21st International Conference on
  Data Engineering (ICDE)}, 2005, pp. 791--802.

\bibitem{tatbul07}
N.~Tatbul, U.~{\c{C}}etintemel, and S.~Zdonik, ``{Staying Fit: Efficient Load
  Shedding Techniques for Distributed Stream Processing},'' in \emph{Proc. 33rd
  international Conference on Very Large Data Bases (VLDB)}, 2007, pp.
  159--170.

\bibitem{chatzistergiou14}
A.~Chatzistergiou and S.~D. Viglas, ``{Fast Heuristics for Near-Optimal Task
  Allocation in Data Stream Processing over Clusters},'' in \emph{Proc. 23rd
  ACM International Conference on Conference on Information and Knowledge
  Management}, 2014, pp. 1579--1588.

\bibitem{khandekar09}
R.~Khandekar, K.~Hildrum, S.~Parekh, D.~Rajan, J.~Wolf, K.-L. Wu, H.~Andrade,
  and B.~Gedik, ``{COLA: Optimizing Stream Processing Applications via Graph
  Partitioning},'' in \emph{Proc. 10th International Middleware Conference
  (Middleware)}, 2009, pp. 308--327.

\bibitem{riabov06}
A.~Riabov and Z.~Liu, ``{Scalable Planning for Distributed Stream Processing
  Systems},'' in \emph{Proc. International Conference on Automated Planning and
  Scheduling (ICAPS)}, 2006, pp. 31--41.

\bibitem{wolf08}
J.~Wolf, N.~Bansal, K.~Hildrum, S.~Parekh, D.~Rajan, R.~Wagle, K.-L. Wu, and
  L.~Fleischer, ``{SODA: An Optimizing Scheduler for Large-Scale Stream-Based
  Distributed Computer Systems},'' in \emph{Proc. 9th International Middleware
  Conference (Middleware)}, 2008, pp. 306--325.

\bibitem{xing06}
Y.~Xing, J.-H. Hwang, U.~{\c{C}}etintemel, and S.~Zdonik, ``{Providing
  Resiliency to Load Variations in Distributed Stream Processing},'' in
  \emph{Proc. 32nd International Conference on Very Large Data Bases (VLDB)},
  2006, pp. 775--786.

\bibitem{pietzuch06}
P.~Pietzuch, J.~Ledlie, J.~Shneidman, M.~Roussopoulos, M.~Welsh, and
  M.~Seltzer, ``{Network-aware Operator Placement for Stream-processing
  Systems},'' in \emph{Proc. 22nd International Conference on Data Engineering
  (ICDE)}, 2006, pp. 49--60.

\bibitem{mattheis12}
S.~Mattheis, T.~Schuele, A.~Raabe, T.~Henties, and U.~Gleim, ``{Work Stealing
  Strategies for Parallel Stream Processing in Soft Real-Time Systems},'' in
  \emph{Proc. International Conference on Architecture of Computing Systems
  (ARCS)}, 2012, pp. 172--183.

\bibitem{hartmann10}
S.~Hartmann and D.~Briskorn, ``{A Survey of Variants and Extensions of the
  Resource-Constrained Project Scheduling Problem},'' \emph{European Journal of
  Operational Research}, vol. 207, no.~1, pp. 1--14, 2010.

\bibitem{wang08}
J.-B. Wang, C.~Ng, and T.~E. Cheng, ``{Single-Machine Scheduling with
  Deteriorating Jobs under a Series-parallel Graph Constraint},''
  \emph{Computers \& Operations Research}, vol.~35, no.~8, pp. 2684--2693,
  2008.

\bibitem{rayward87}
V.~J. Rayward-Smith, ``{UET Scheduling with Unit Interprocessor Communication
  Delays},'' \emph{Discrete Applied Mathematics}, vol.~18, pp. 55--71, 1987.

\bibitem{finta96}
L.~Finta, Z.~Liu, I.~Mills, and E.~Bampis, ``{Scheduling UET-UCT
  Series-Parallel Graphs on Two Processors},'' \emph{Theoretical Computer
  Science (TCS)}, vol. 162, no.~2, pp. 323--340, 1996.

\end{thebibliography}
